\documentclass{article}
\pdfoutput=1
\usepackage[T1]{fontenc}
\usepackage[utf8]{inputenc}

\usepackage[colorinlistoftodos]{todonotes}

\usepackage{amsmath,amsthm,amssymb,fullpage,listings,color}
\usepackage{xspace}
\usepackage{authblk}
\usepackage[hyphens]{url}
\usepackage{hyperref}
\usepackage[hyphenbreaks]{breakurl}

\definecolor{darkblue}{HTML}{00008B} 
\hypersetup{
    colorlinks=true,
    urlcolor=darkblue,
    citecolor=black,
    linkcolor=black,
}

\definecolor{keywordcolor}{rgb}{0.7, 0.1, 0.1}   
\definecolor{commentcolor}{rgb}{0.4, 0.4, 0.4}   
\definecolor{symbolcolor}{rgb}{0.0, 0.1, 0.6}    
\definecolor{sortcolor}{rgb}{0.1, 0.5, 0.1}      

\newcommand{\code}{\lstinline}

\newcommand{\abs}[1]{\left\lvert{#1}\right\rvert}
\newcommand{\mathlib}{{\fontfamily{lmss}\selectfont mathlib}\xspace}
\newcommand{\R}{\mathbb{R}}
\newcommand{\Q}{\mathbb{Q}}

\DeclareMathOperator{\id}{id}
\DeclareMathOperator{\Aut}{Aut}

\newtheorem{theorem}{Theorem}
\newtheorem{proposition}[theorem]{Proposition}
\newtheorem{lemma}[theorem]{Lemma}
\theoremstyle{definition}
\newtheorem{definition}[theorem]{Definition}

\title{Formalizing Galois Theory}
\author[1]{Thomas Browning}
\author[1]{Patrick Lutz}
\affil[1]{University of California, Berkeley}

\begin{document}

\maketitle

\begin{abstract}
We describe a project to formalize Galois theory using the Lean theorem prover, which is part of a larger effort to formalize all of the standard undergraduate mathematics curriculum in Lean. We discuss some of the challenges we faced and the decisions we made in the course of this project. The main theorems we formalized are the primitive element theorem, the fundamental theorem of Galois theory, and the equivalence of several characterizations of finite degree Galois extensions.
\end{abstract}

\section{Introduction}


The Lean theorem prover (a proof assistant developed by Leonardo de Moura and collaborators at Microsoft research \cite{demoura2015lean}) was first released in 2013. But the past couple of years have witnessed an explosion of work on formalizing math in Lean. This explosion has included the formalization of some surprisingly advanced topics in mathematics, including a proof of the independence of the continuum hypothesis \cite{han2019formalization}, the theory of Witt vectors \cite{commelin2020formalizing}, and the definition of perfectoid spaces \cite{buzzard2020formalizing}.

These accomplishments have been supported by the formalization of a large body of more prosaic definitions and results, which the advanced topics depend on. Much of this work has been done as part of the development of \mathlib, a large open-source library of math formalized in Lean \cite{mathlib2020lean}. In fact, the development of \mathlib and the community around it is one of the main factors that led to the recent explosion of math formalized in Lean. And \mathlib has continued to grow rapidly. In the past year, more than $250,000$ lines of code have been added, including proofs of the Cayley-Hamilton theorem, Carath\'{e}odory's theorem, and the fundamental theorem of calculus. It now contains a substantial percentage of the standard undergraduate math curriculum. 

We will report on a project to formalize Galois theory in Lean, which was undertaken as a contribution to \mathlib, filling one of the remaining holes in its coverage of undergraduate math. This paper can be seen as a case study in what it's currently like to formalize a significant piece of undergraduate-level mathematics using the Lean and \mathlib. We will give an overview of how we formalized the proofs of some of the theorems of Galois theory and comment on the obstacles that we encountered, the design choices we made, and the insights into the original mathematics that were revealed by formalization.

\subsection{History of This Project}

The project to formalize Galois theory in Lean was started by a group of undergraduate students at Imperial College London. At that time, the algebra library of Lean was less developed and so much of their work went into building up basic definitions and facts about rings, fields, and algebras. Kenny Lau, in particular, formalized many essential definitions and proved several key results needed to formalize Galois theory, including facts about minimal polynomials, degrees of field extensions, and fixed fields of groups acting on fields. He also went on to formalize several other fundamental theorems in field theory, such as the existence of splitting fields and algebraic closures.

In the summer of 2020, we organized a seminar at Berkeley to learn how to use the Lean theorem prover. We began working on formalizing Galois theory as an extension of that seminar, building on the work by the group at Imperial to finish the proof of the fundamental theorem of Galois theory. In doing so, we benefited from the extensive work that has gone into \mathlib's algebra library, including many contributions by Anne Baanen, a few of which we will mention below. In turn, the Galois theory that we formalized is now part of \mathlib and will continue to be added to and modified as part of \mathlib's ongoing development. In fact, Baanen and others have already made use of some of our work in their formalization of Dedekind domains and ideal class groups \cite{baanen2021formalization}.

The history of this project demonstrates an important point about how the \mathlib community functions and how it has been able to grow so rapidly. Namely, the organization and open-source nature of \mathlib enables collaboration from hundreds of people around the world, including both professionals working on formalization full-time, as well as undergraduates (like the Imperial group) and graduate students (like us) who have picked up Lean out of curiosity. Our project would have taken many times longer to complete if we had not been able to build off the work of Lau, Baanen, and many others. This feature of \mathlib development is also discussed in a paper by Baanen, Dahmen, Narayanan, and Nuccio on formalizing algebraic number theory in Lean \cite{baanen2021formalization}.

\subsection{Other Formalizations of Field Theory and Galois Theory}

This is not the first time that Galois theory, or field theory more generally, has been formalized. We will mention a few notable examples here, though there may be others we are unaware of.

A number of results in field theory and Galois theory have been formalized in Coq as part of the Mathematical Components library. Some of these results, including all of the results that we discuss in this paper, as well as several that have not yet been formalized in Lean, were first formalized more than ten years ago as part of the project to formalize the Feit--Thompson odd order theorem \cite{gonthier2013machine}. More recently, Sophie Bernard, Cyril Cohen, Assia Mahboubi, and Pierre-Yves Strub completed a project to formalize a proof of the Abel-Ruffini theorem on the insolvability of the quintic \cite{bernard2021unsolvability}. At the end of section \ref{section:algebra}, we will comment on a couple specific implementation differences between the development of Galois theory in Coq and in Lean.

The only other attempt to formalize Galois theory that we are aware of is the GALOIS project at Manchester University in the 1990s to formalize Galois theory in the LEGO proof assistant \cite{aczel1994notestowards}. However, this project was never completed and to the best of our knowledge there was never a significant amount of Galois theory formalized in LEGO.

More generally, results in field theory have been formalized in several other proof assistants, including Isabelle, Mizar, and HOL Light. The development of field theory in Isabelle is especially advanced---for example, the existence of algebraic closures of fields was recently proved in Isabelle by Paulo Em\'{i}lio de Vilhena and Lawrence Paulson \cite{devilhena2020algebraically}. Also, the impossibility of trisecting an angle with a ruler and compass has been proved in both Isabelle and HOL Light (though notably the proof in HOL Light avoids talking about field extensions) and a number of results about fields have been formalized in Mizar, including the construction of a field extension adjoining a root of a given polynomial \cite{schwarzweller2020ring}.


\subsection{What Did We Prove?}

We formalized three main results in Galois theory.

\begin{theorem}[Primitive Element Theorem]
Let $E/F$ be a separable field extension of finite degree. Then $E = F(\alpha)$ for some $\alpha\in E$.
\end{theorem}


\begin{theorem}[Fundamental Theorem of Galois Theory]
If $E/F$ is a Galois extension of finite degree then there is an inclusion-reversing bijection between the intermediate fields of $E/F$ and the subgroups of the automorphism group $\Aut(E/F)$.
\end{theorem}

\begin{theorem}[Equivalent Characterizations of Galois Extensions]
Let $E/F$ be a field extension of finite degree. The following are equivalent:
\begin{enumerate}
    \item $E/F$ is Galois, i.e.\ separable and normal.
    \item $F$ is the fixed field of $\Aut(E/F)$.
    \item $\abs{\Aut(E/F)} = [E : F]$.
    \item $E$ is the splitting field of a separable polynomial.
\end{enumerate}
\end{theorem}

As we mentioned earlier, our formalizations of these results are now part of \mathlib. This means they will continue to be maintained and built upon rather than gradually becoming incompatible with other developments in Lean, and also that their statements can be viewed in the online \mathlib documentation (the relevant pages in the documentation are \href{https://leanprover-community.github.io/mathlib_docs/field_theory/galois.html}{\texttt{galois.lean}} and \href{https://leanprover-community.github.io/mathlib_docs/field_theory/primitive_element.html}{\texttt{primitive\_element.lean}}). On the other hand, it also means that eventually the contents of this paper may not exactly match the formalized statements online. Therefore we have chosen a specific commit of the \mathlib repository on Github to use as a reference for the formalized versions of the statements we discuss in this paper. The commit we have chosen is from January 1st, 2021 and can be found at this url: \burl{https://github.com/leanprover-community/mathlib/tree/9535f9148b91e33d12281ef5ac99292c2fb62220/src/field_theory}. This commit provides a snapshot of \mathlib at one point in time and is guaranteed to be compatible with the development of Galois theory as we present it in this paper, no matter what changes occur in later versions of \mathlib. Since it is hosted by Github rather than any individual, it is also likely to remain available online for many years.



\subsection{What Did We Learn?}

The process of formalizing a piece of mathematics essentially involves a translation between two mathematical languages. The first is the language in which people usually do mathematics, which includes centuries' worth of conventions and heuristics and draws on the powerful inference abilities of the human brain. The second is a language in which a computer can understand mathematics---in our case, Lean. We will sometimes refer to the first language as ``human mathematics'' in contrast to ``formalized mathematics.''

Many things are easier in human mathematics than in formalized mathematics and a few things are harder. Most things are at least a bit different. In this paper we will try to explain the differences between human mathematics and formalized mathematics that we encountered in the process of translation. Our intent is for our explanation to be understandable by a mathematician with no experience in formalization.

Some of the differences that we encountered are inherent to the process of formalization no matter what proof assistant is being used for formalization and some are caused by features of Lean that are not shared in all proof assistants. The most important of these features is that Lean is built on type theory rather than set theory (and in particular, a special kind of type theory called ``dependent type theory''). Lean shares this feature with many other proof assistants, including Coq (which also uses dependent type theory) and Isabelle/HOL (which uses simple type theory rather than dependent type theory). The main thing to know about type theory is that every object is of exactly one type. This means that when introducing new objects, we have to explicitly decide what type they are.

In human mathematics, things are different: it is normal to view an object from multiple perspectives and it is so easy that we often don't realize we are doing it. In type theory, it is also possible to view an object in multiple ways, but to do so we need to define a map, called a \emph{coercion}, between the two types. For example, if we want to think of a subgroup of a group $G$ as being a subset of $G$ then we need to explicitly define the map from subgroups of $G$ to subsets of $G$ and then insert this map everywhere that we want to think of a subgroup as a subset. Dealing with coercions may sound arduous, but fortunately Lean has automation to help with it; we still need to explicitly define the map, but in many cases Lean can automatically infer where it should be inserted. As we will see later on, being forced to explicitly think about these maps can have beneficial side-effects.

Formalized math also forces us to think more explicitly about which algebraic structures can be put on a mathematical object and how the relationships between those algebraic structures should be expressed. We will consider the real numbers as a motivating example.

In Lean, we can express the fact that the real numbers are a field as follows. The collection of real numbers is a type (the elements of that type are exactly the real numbers) and for every type $X$, there is a type of all field structures on $X$. Describing a field structure on the real numbers consists of constructing an element of the type of field structures on the type of real numbers.

Note that the type of field structures on the type of real numbers may have many different elements---this corresponds to the fact that we can put many different field structures on the real numbers. What we mean in human mathematics when we say that the set of real numbers ``is'' a field is really that one of these choices is canonical and should be used in almost all cases. To accommodate this type of situation, Lean provides a way to say that one specific element of the type of field structures on a type $X$ should be automatically used in any situation where we need a field structure on $X$.

In human mathematics, the fact that the real numbers are a field implies other facts about the real numbers. For example, that they are a division ring. Lean also has automation to help deal with this. Essentially, we can define a function from field structures on a type $X$ to division ring structures on $X$ and Lean can then automatically apply the function whenever it is needed. This ability is known as \emph{type class inference}.

Something analogous to type class inference often happens in human mathematics as well, but it can often be done on the fly and without needing to be spelled out explicitly anywhere. For example, if you are told that $F$ is an intermediate field of $\R/\Q$, then it is immediately clear that $\R$ is a vector space over $F$. But in Lean, we need to at some point explicitly show that this inference is possible (or at least we need to provide Lean with other inferences which imply this one---Lean is quite good at chaining these sorts of inferences together).

Even though type class inference is automatic in Lean, we still have to be explicit about exactly what sorts of algebraic structure we need to have in the statements of lemmas as well as within proofs. Returning to the scenario of the previous paragraph, if we need to know in the middle of a proof that $\R$ is a vector space over $F$ then we should make sure that we have already supplied Lean with the necessary facts to infer this.

All of this means that it is important to set up a hierarchy of structures and inferences so that Lean's automation is able to infer as much as possible without too much manual intervention. Sometimes this requires creating new structures that don't have any exact analogue in human mathematics. One structure in Lean that turned out to be particularly useful in our project is something called \code{is_scalar_tower}, which essentially expresses the relationship between a tower of field extensions $E/L/F$, but in a bit more generality. We will discuss this structure in more detail later in the paper.

Here's one general takeaway from all of this. Working in formalized mathematics and especially in type theory can sometimes be annoying because you are forced to pick a type for everything and explicitly construct coercions and instances of structures that might be so obvious that they don't need to be mentioned in human mathematics. On the other hand, this can actually be positive because it encourages us to think in more structural terms and notice abstract patterns that are repeated in different situations. Sometimes this even leads to simpler and more efficient proofs. For example, after defining what it means to adjoin an element to a field, we were able to use the theory of Galois insertions (which had already been developed in \mathlib) to immediately infer that the intermediate fields of a field extension form a complete lattice. Later in this paper, we will explain this example in more detail and also explain why the formalization process made the connection to Galois insertions clearer.

\section{Mathematical Overview}

Before we go further, we will review some concepts from Galois theory. This section is entirely in the language of human mathematics and does not contain any comments on the formalization process.

A polynomial $p$ over a field $F$ is said to be \textit{separable} if it is coprime with its formal derivative. This ensures that it has no repeated roots in any field extension of $F$.
If $E/F$ is an algebraic field extension then each element of $E$ has a minimal polynomial over $F$.
A field extension $E/F$ is \textit{separable} if the minimal polynomial of each element of $E$ is separable. A field extension $E/F$ is \textit{normal} if the minimal polynomial of each element of $E$ splits into linear factors over $E$.
A field extension $E/F$ is \textit{Galois} if it is both separable and normal.
The group $\Aut(E/F)$ is the group of automorphisms of $E$ which fix $F$ pointwise.
When $E/F$ is a Galois extension, this group is called the Galois group of $E$ over $F$.

The intuitive idea behind this definition of a Galois extension is that all the automorphisms of $E$ which fix $F$ pointwise come from permuting the roots of polynomials over $F$, and the condition that $E/F$ is a Galois extension ensures that we are not missing any roots, and hence that we are also not missing any automorphisms. This concept of ``not missing any automorphism'' can be made precise as follows. For any field extension $E/F$ of finite degree, we have the inequality $\abs{\Aut(E/F)} \leq [E : F]$. The condition that $E/F$ is Galois is equivalent to the statement that this inequality is actually an equality.

When an extension $E/F$ is Galois, there is a close correspondence between the intermediate fields of $E/F$ and the subgroups of $\Aut(E/F)$. First, we can associate any intermediate field $L$ of $E/F$ to the subgroup $\Aut(E/L)$ of $\Aut(E/F)$ consisting of those automorphisms of $E$ which fix every element of $L$. Second, we can associate any subgroup $H$ of $\Aut(E/F)$ to the intermediate field $E^H$ of $E/F$ consisting of those elements of $E$ which are fixed by every element of $H$.
For any extension $E/F$, these give inclusion-reversing maps between intermediate fields of $E/F$ and subgroups of $\Aut(E/F)$. When $E/F$ is Galois, these maps are inverse bijections. The existence of this inclusion-reversing bijection between intermediate fields of a Galois extension $E/F$ and subgroups of $\Aut(E/F)$ is referred to as the fundamental theorem of Galois theory or the Galois correspondence.


In this paper, and in the formalization project on which it is based, we follow a slightly unconventional route to proving the fundamental theorem of Galois theory. Namely, we make use of the primitive element theorem (which is usually treated essentially independently from the fundamental theorem of Galois theory). This theorem states that if $E/F$ is a separable field extension of finite degree then $E$ is generated over $F$ by a single element---i.e.\ there is some $\alpha \in E$ such that $E = F(\alpha)$, where $F(\alpha)$ denotes the smallest subfield of $E$ containing both $F$ and $\alpha$. We will explain in section \ref{section:ftgt} how we use the primitive element theorem in proving the fundamental theorem of Galois theory.

\section{The Primitive Element Theorem}

In this section, we will explain some of the things we encountered while formalizing the proof of the primitive element theorem.
An apparently minor point is that in order to even state the primitive element theorem, we need to first define what it means to adjoin an element to a field. In human mathematics, this is a relatively simple definition. But to formalize it, there are several important design decisions that must be made. 

These decisions took trial-and-error to get right and it turned out that making them in a good way had significant benefits later on in our project, including in our proof of the primitive element theorem. They also serve as helpful illustrations of some general points about the process of translating human mathematics into formal mathematics. For these reasons, the way that we made these decisions will be the main focus of this section.

\subsection{Field Extensions}
\label{section:algebra}
Before we discuss adjoining elements to fields, we need to briefly explain how \mathlib handles fields, field extensions, and towers of fields.

Fields in \mathlib are essentially the same as in human mathematics. In human math, we say ``$F$ is a field'' and we mean that $F$ consists of a set together with some operations on that set which satisfy the field axioms. In Lean, we write \code{(F : Type*) [field F]}, which means that $F$ is a type and we have an element of the type of field structures on $F$. Here, \code{field F} refers to the type of all field structures on $F$.

Field extensions in \mathlib, however, are somewhat different from human mathematics, but this is partly due to a useful ambiguity in human math. In human mathematics, the phrase ``$E/F$ is a field extension'' may mean that $F$ is literally a subset of $E$ or it may mean that $F$ embeds into $E$ (with the embedding usually being inferred from context). In fact, it is standard to switch between these two options whenever it is convenient. However, when formalizing the the definition of a field extension, we must pick one of these two meanings and which one we pick has consequences for how later definitions and proofs are formalized.



The course chosen by \mathlib is to use a version of the second option.
More precisely, in \mathlib, the way to express that $E/F$ is a field extension is to say that $E$ is a field which is also an $F$-algebra, which in \mathlib just means that there is a ring homomorphism\footnote{In general, \mathlib defines an $R$-algebra to be a semiring $A$ together with a ring homomorphism $f \colon R \to A$ such that for all $r \in R$ and $x \in A$, $f(r)$ commutes with $x$. In particular, algebras are associative and unital.} $F \to E$.
The benefit of this definition of a field extension is that it is more general than requiring $F$ to be a subset of $E$, and it allows lemmas about algebras and vector spaces to also be applied to field extensions.
To define a field extension in Lean, we write:
\begin{lstlisting}
(F E : Type*) [field F] [field E] [algebra F E]
\end{lstlisting}
which should be read as ``$F$ and $E$ are types, there are field structures on $F$ and $E$, and there is an $F$-algebra structure on $E$''.


The downside of this approach is that it can be cumbersome to handle situations where we have multiple different field extensions of $F$ which are related to each other in some way. The most common situation of this sort is when we have a tower of field extensions $K/E/F$. In human math, it is common to talk about such a tower by assuming that $F$ and $E$ are both subsets of $K$. But because of the way field extensions are defined in \mathlib, we can't do this. Instead, to keep things manageable, \mathlib has the structure \code{is_scalar_tower}.


The definition of \code{is_scalar_tower} is very general (which is part of its strength). It abstracts not just towers of field extensions, but more general sorts of towers as well. First, if $X$ and $Y$ are types then a scalar action of $X$ on $Y$ is just a map $\bullet \colon X\times Y \to Y$. Note that an $F$-algebra structure on a field $E$ is a special case of this. Now if $X$, $Y$, and $Z$ are types and we have scalar actions $\bullet \colon X \times Y \to Y$, $\bullet \colon X \times Z \to Z$, and $\bullet\colon Y\times Z \to Z$ then \code{is_scalar_tower} asserts the axiom $(x \bullet y)\bullet z = x\bullet(y\bullet z)$ for any elements $x$ of $X$, $y$ of $Y$, and $z$ of $Z$.
In the case of a tower of fields $K/E/F$, the \code{is_scalar_tower} axiom is just equivalent to the compatibility of the ring homomorphisms $F\to E\to K$ and $F\to K$.
To define a tower of field extensions in Lean, we write
\begin{lstlisting}
(F E K : Type*) [field F] [field E] [field K] [algebra F E] [algebra F K] [algebra E K]  [is_scalar_tower F E K]
\end{lstlisting}
This is a bit of a mouthful, but luckily you only have to write it once at the top of a file, and then you can use the setup throughout the file.

We believe that handling field extensions and towers of field extensions is beneficial because it allows lemmas to be stated and proved in more generality and to be reused in more contexts (we will see one example of this later) and that this is especially beneficial in the context of a project like \mathlib which aims to produce a large coherent library of formalized math. However, this extra generality is sometimes cumbersome to work with, since statements and proofs become more involved due to the need to talk about compatibility of algebra structures. And though structures like \code{is_scalar_tower} and Lean's automation help, they do not completely eliminate the inconvenience. 

Partly for these reasons, the authors of Coq's Mathematical Components library (a.k.a.\ mathcomp) chose to define field extensions differently. In their development of Galois theory, they first pick a large ambient field $L$ and most of the time only work with subfields of $L$. So a field extension $E/F$ just means a pair of subfields $E$ and $F$ of $L$ such that $F \subseteq E$. This means there is no need to check ``compatibility'' of the field extensions in a tower---they are always compatible because the embeddings are always just inclusion maps. Getting rid of such compatibility checks also reduces the compilation time of the formalized proofs. On the other hand, it is occasionally necessary to leave the field $L$ and this means that the mathcomp approach cannot totally avoid the difficulties of the \mathlib approach.

It seems that overall the differences between the mathcomp and \mathlib approaches to defining field extensions are not so much consequences of differences between Coq and Lean, but rather of differences between the goals and conventions of the two communities. The design decision of working with substructures of a large ambient structure is common in mathcomp. For example, one of the main formalization efforts that used mathcomp (and which drove a lot of the development of mathcomp) was the formalization of the Feit-Thompson theorem of finite group theory. That project followed a similar strategy of mostly working with subgroups of some large ambient group. In contrast, working in the setting of structures with maps between them (and using structures like \code{is_scalar_tower} to keep track of the complexity) is common in \mathlib.

\subsection{Adjoining Elements}


Mathematically, if $E/F$ is a field extension then adjoining elements $S\subseteq E$ to the base field $F$ is straightforward. Our definition of $F(S)$ in Lean follows the usual mathematical definition of taking the closure of $F\cup S$ under the field operations. However, there are some interesting structural decisions that we are forced to make.

Since Lean uses type theory, everything we define must have a type. So what should we choose for the type of $F(S)$? There are many options:
\begin{itemize}
    \item \code{set E}, the type of subsets of $E$,
    \item \code{subfield E}, the type of subfields of $E$,
    \item \code{subalgebra F E}, the type of $F$-subalgebras of $E$,
    \item \code{intermediate_field F E}, the type of $F$-subalgebras of $E$ which are also subfields of $E$.
\end{itemize}

When we started working on this project, the types \code{intermediate_field F E} and \code{subfield E} did not exist.
Instead, there was a predicate \code{is_subfield} on subsets of $E$. From the other two options, we initially chose \code{set E}, which turned out to be the source of a few headaches. For example, we had to manually show how to convert $F(S)$ into an $F$-submodule of $E$. And often lemmas had to be stated with additional assumptions that could have been inferred directly if we had chosen a different type.
For example, the following lemma states that an intermediate field of $E/F$ contains $F(S)$ if and only if it contains $S$, but requires the ugly assumption \code{HF}:
\begin{lstlisting}
lemma adjoin_le_iff {K : set E} [is_subfield K] (HF : set.range (algebra_map F E) ⊆ K) : adjoin F S ⊆ K ↔ S ⊆ K
\end{lstlisting}

To alleviate some of these issues, we eventually switched to using \code{subalgebra F E}. Later, Anne Baanen created the type \code{intermediate_field F E}, which turned out to be the best choice for the type of $F(S)$. It ensures that $F(S)$ ``remembers'' that it is both a field extension of $F$ and a subfield of $E$, which allows us to automatically invoke lemmas and theorems that depend on these assumptions. For example, we no longer have to explicitly show that $F(S)$ is an $F$-submodule of $E$---this fact is simply inherited from the fact that any \code{intermediate_field F E} is an $F$-submodule of $E$. It also allows us to simplify the statements of many lemmas. For example, the lemma stated above can now become: 
\begin{lstlisting}
lemma adjoin_le_iff {K : intermediate_field F E} : adjoin F S ≤ K ↔ S ⊆ K
\end{lstlisting}
Finally, there is no real disadvantage to using \code{intermediate_field} compared to any of the alternatives; it is always possible to turn an \code{intermediate_field F E} into a \code{subfield E} or a \code{subalgebra F E} if the need ever arises.

What we learned is that, as a general rule of thumb, it is always best to give mathematical objects the most specific type available and that sometimes it is necessary to create a new type if the none of the available ones are specific enough. The reason is that this ensures that no information is lost---we can always convert a more specific type into a less specific one, but not vice-versa. As we saw when we tried using \code{set E} as the type of $F(S)$, if we choose a type which is too general then we will often need to insert extra assumptions in the statements of lemmas. This leads to unnatural looking statements which are harder to work with. When the right type is chosen, these assumptions become implicit in the type of the object, leading to statements that look closer to their human mathematics equivalents.

The creation of the type \code{intermediate_field F E} illustrates a general point about formalizing math, which will show up several more times in this paper: it is important to find abstractions which hide the right amount of information---especially information which is normally left implicit in human mathematics. Otherwise, it is easy for formalized statements to quickly become too complicated and hard for humans to parse.




\subsection{The Lattice Structure of Intermediate Fields}
Whenever you define something new in Lean, the first thing to do is to prove lots of little lemmas about how the definition interacts with previously defined notions. Sometimes this process leads to a new and simpler perspective.
The lemma \code{adjoin_le_iff} from the previous section provides an example of this: an intermediate field of $E/F$ contains $F(S)$ if and only if it contains $S$. In human mathematics, this is expressed by the following (obvious) lemma.
\begin{lemma}
\label{lem:adjoin_le_iff}
Let $E/F$ be a field extension, $S$ a subset of $E$, and $K$ an intermediate field of $E/F$. Then $F(S)\subseteq K$ if and only if $S\subseteq K$.
\end{lemma}
Earlier, we translated this statement into Lean as:
\begin{lstlisting}
lemma adjoin_le_iff {K : intemediate_field F E} : adjoin F S ≤ K ↔ S ⊆ K
\end{lstlisting}
However, if we are not careful then Lean will complain about this statement. The problem is that $S \subseteq K$ doesn't seem to make sense since $S$ has type \code{set E} but $K$ has type \code{intermediate_field F E}. When a human reads $S \subseteq K$ they immediately infer that an intermediate field of $E/F$ can also be thought of as a subset of $E$, but Lean needs more help. At first, this seems purely like an annoyance. But as we will see, dealing with it actually helps make clear a pattern which occurs frequently in mathematics and which we can exploit to automatically infer many basic facts about intermediate fields.

The solution to the problem is to realize that the conversion of an intermediate field of $E/F$ into a subset of $E$, which humans can perform effortlessly, actually consists of applying a function from intermediate fields to subsets (in category theoretic language, this is the forgetful functor from intermediate fields to subsets). In Lean, this function is called a \textit{coercion} and it needs to be manually defined. Once the coercion has been defined, Lean will be able to automatically apply it in appropriate situations and will no longer complain when we write things like $S \subseteq K$.

Having to explicitly provide Lean with the coercion highlights the fact that Lemma \ref{lem:adjoin_le_iff} should really be viewed as a relationship between the adjoining map \code{set E → intermediate_field F E} and the coercion \code{intermediate_field F E → set E}. In fact, this relationship is very common in mathematics. The definition below is one way to capture this relationship abstractly (which some mathematicians might recognize as a special case of adjoint functors).



\begin{definition}
Let $P$ and $Q$ be partially ordered sets.
A \textit{Galois connection} between $P$ and $Q$ consists of two order-preserving functions $f\colon P\to Q$ and $g\colon Q\to P$ such that $f(p)\leq q$ if and only if $p\leq g(q)$.
A \textit{Galois insertion} of $Q$ into $P$ is a Galois connection that satisfies $f\circ g=\id_Q$.\footnote{The reason that Galois' name appears here is that the Galois correspondence is an example of a Galois insertion.}
\end{definition}

Thus, we see that Lemma \ref{lem:adjoin_le_iff} states that the adjoining map \code{set E → intermediate_field F E} and the coercion \code{intermediate_field F E → set E} form a Galois connection.
In fact, it is not hard to see that they form a Galois insertion of \code{intermediate_field F E} into \code{set E}.

\begin{lemma}
Let $E/F$ be a field extension. The map $S \mapsto F(S)$ and the forgetful functor from intermediate fields of $E/F$ to subsets of $E$ form a Galois insertion of the partial order of intermediate fields of $E/F$ into the partial order of subsets of $E$.
\end{lemma}

Proving this lemma allows us to automatically import all the facts in \mathlib that have already been proved about Galois insertions. For example, the following lemma is in \mathlib.
Recall that a complete lattice is a partially ordered set in which every subset has both a greatest lower bound and a least upper bound.


\begin{lemma}
\label{lem:complete_lattice}
Suppose that there is a Galois insertion of $Q$ into $P$. If $P$ is a complete lattice then so is $Q$.
\end{lemma}

Applying this lemma shows that \code{intermediate_field F E} is a complete lattice. This in turn gives us many notions for free. For example, it gives us the bottom intermediate field (denoted $\bot$), and the top intermediate field (denoted $\top$), as well as intersections and compositums of intermediate fields. It also gives us access to any fact about complete lattices that is contained in \mathlib. This is amazing since all we had to do, more or less, was provide the definition of adjoining elements to fields and prove Lemma \ref{lem:adjoin_le_iff}---both of which we were going to do anyways.

This trick of pulling back lattice structure along Galois insertions was not discovered by us.
We learned it from looking at Kenny Lau's work on subalgebras in \mathlib.
The trick is also used in \mathlib to get the lattice structure on partitions and on a topological space.

This section has highlighted a particular example of a larger theme, which is that formalizing math (especially in a system based on type theory, like Lean) forces you to make explicit many things that are normally left implicit in human mathematics. Sometimes this is merely an annoyance, but at other times it can actually be helpful because it makes you notice structural patterns which are repeated in many different situations. These patterns can be abstracted out into highly versatile theorems which can then be invoked in all similar situations.


\subsection{An Induction Scheme for Intermediate Fields}

One reason that it is useful to formalize the definition of adjoining elements to a field is that it allows results to be proved inductively by adjoining one element at a time. Here's how this typically works in human mathematics: we start with a field extension $E/F$ of finite degree. We then pick some finite sequence of elements $\alpha_1, \alpha_2, \ldots, \alpha_n$ such that $F(\alpha_1, \alpha_2, \ldots, \alpha_n) = E$ and prove by induction on $n$ that whatever property we are interested in holds for $F(\alpha_1, \alpha_2, \ldots, \alpha_n)$.

It is possible to carry out this sort of proof in Lean, but it usually does not feel very natural. In particular, it can often be a little difficult to work with arbitrary finite sequences of elements. One way to avoid having to deal with these difficulties is to formulate an induction principle for intermediate fields. The idea is that if we can abstract out the requirements necessary to do the kind of induction we described above, we can avoid having to actually work with finite sequences of elements (except in the proof of the abstract induction principle itself). This maneuver is mostly just a way to avoid having to deal with certain types of arguments that don't work very smoothly in Lean, but it does have the added benefit that some standard proofs become simpler when rewritten using this new induction principle.

Here's the induction principle.

\begin{lemma}
\label{lem:induction1}
Let $E/F$ be a field extension, $S$ a finite subset of $E$, and $P$ be a predicate on intermediate fields of $E/F$. Suppose that $P$ holds of $F$ and that for all intermediate fields $K$ of $E/F$ and all $\alpha\in S$, if $P$ holds of $K$ then $P$ holds of $K(\alpha)$. Then $P$ holds of $F(S)$.
\end{lemma}

In the next section, we will see how to use this induction principle to prove the primitive element theorem. But first, let's point out one issue in formalizing Lemma \ref{lem:induction1} in Lean. The lemmas depend on the following fact.

\begin{lemma}
Let $E/F$ be a field extension and let $S,T\subseteq E$. Then $F(S)(T) = F(S\cup T)$.
\end{lemma}

The trouble with this fact is that the field $F(S\cup T)$ has type \code{intermediate_field F E}, but the field $F(S)(T)$ has type \code{intermediate_field F(S) E} and in type theory you can only say that two objects are equal if they have the same type. The solution is the same as the one described in the previous section: we need to define a coercion from \code{intermediate_field K E} to \code{intermediate_field F E} whenever $K$ is an intermediate field of $E/F$. In fact, we also need this coercion when formalizing the statement of the induction lemmas in order to make sense of the phrase ``$P$ holds of $K(\alpha)$'' (recall that $P$ was a predicate on intermediate fields of $E/F$ whereas $K(\alpha)$ is an intermediate field of $E/K$).

We will note here that instead of using the induction scheme we have described in this section, it is often possible to instead use induction on the degree of the field extension. One reason we avoided this in our work was that it is a bit messy in Lean (for technical reasons having to do with the way universes are handled in Lean). However, the team that developed Galois theory in Coq did use this approach sometimes. They did not run into the technical inconveniences we encountered with this approach in part because they were working with subfields of a single large ambient field, rather than with arbitrary fields, as we did. Both ways of carrying out such induction arguments seem to work fine, but as we have already mentioned, once we had formulated the induction scheme explained here, it did seem to make a number of proofs smoother.

\subsection{The Primitive Element Theorem}

In this section we will explain our formalization of the proof of primitive element theorem.

\begin{theorem}[Primitive Element Theorem]
Let $E/F$ be a separable field extension of finite degree. Then $E = F(\alpha)$ for some $\alpha\in E$.
\end{theorem}

Let's first review the usual proof of this theorem, following the proof in Milne's \emph{Fields and Galois Theory} textbook \cite{milne2020fields}. The proof splits into two cases depending on whether the fields are finite or infinite. If the fields are finite then $\alpha$ can be taken to be the generator of the finite cyclic group $E^\times$. If the fields are infinite then the theorem can be proved by induction on $[E : F]$. The inductive step relies on the following key lemma.

\begin{lemma}
\label{lem:inductive_step}
If $F$ is infinite and if $E/F$ is a separable field extension then for all $\alpha,\beta\in E$, there exists a $\gamma\in E$ such that $F(\alpha, \beta) = F(\gamma)$.
\end{lemma}

To apply this lemma, we first show that if $[E : F] > 1$ then there is some element $\alpha\in E$ which is not in $F$. If $\alpha$ generates $E$ over $F$ then we are done. If not, then we can show that $[E : F(\alpha)] < [E : F]$. Then we can invoke the inductive hypothesis to write $E=F(\alpha)(\beta)$ for some $\beta\in E$, and we can apply Lemma \ref{lem:inductive_step} to finish.

One problem with formalizing this proof is that if we do induction on $[E : F]$ then the inductive hypothesis includes a quantifier over all field extensions of $K/L$ of degree less than $[E : F]$. In Lean, this is problematic because we are only allowed to quantify over elements of some specific type and there is no type of all field extensions\footnote{For type theory aficionados: Lean's type theory uses an infinite (non-cumulative) hierarchy of universes. So we \emph{can} form the type of all field extensions that are in any fixed universe, but we prefer to prove the theorem without putting any restrictions on the universes of $F$ and $E$.} (just like in set theory there is no set of all field extensions\footnote{In set theory, however, we are allowed to quantify over this collection even though it is not a set.}).

There are ways around this problem, but it turns out that in this case, a cleaner solution is to replace the induction on the degree of $E/F$ with the induction scheme we introduced in the last section. More precisely, let $P$ be the predicate on intermediate fields of $E/F$ defined by:
\[
P(K) \iff K = F(\alpha) \text{ for some } \alpha \in E.
\]
If we apply the induction scheme of Lemma \ref{lem:induction1} to this predicate then we see that the base case is easy to prove (it is just asking us to show that $F$ can be generated over $F$ by one element) and the inductive step is given by Lemma \ref{lem:inductive_step}.

Switching to this induction scheme has the additional benefit of streamlining the proof in other ways. For example, when proving the inductive step, we do not need the degree to decrease so it is no longer necessary to break into cases depending on whether $\alpha$ generates $E$ over $F$. Also, we no longer need to explicitly show that $E$ contains some element not in $F$---this is essentially baked into the induction scheme.

We will conclude this section by pointing out something interesting about the way that the primitive element theorem is stated in Lean. Here's what it looks like.
\begin{lstlisting}
theorem exists_primitive_element [finite_dimensional F E] (F_sep : is_separable F E) :
  ∃ α : E, F⟮α⟯ = ⊤
\end{lstlisting}
Note that instead of asserting that $F(\alpha) = E$, the theorem says that $F(\alpha) = \top$ (read ``top''). This is because $F(\alpha)$ has type \code{intermediate_field F E} but $E$ does not, and thus they cannot be equal. So instead we prove that $F(\alpha)$ is equal to the largest intermediate field of $E/F$ (which a human mathematician would just automatically identify with $E$). In some other situations, it is necessary to make a similar distinction between $F$ and $\bot$ (read ``bottom''), the smallest intermediate field of $E/F$.

\subsection{Adjoining One Element}

In order to leverage the primitive element theorem in our development of Galois theory, we will need the following two lemmas.

\begin{lemma}
\label{lem:adjoin_simple_degree}
Let $E/F$ be a field extension. If $\alpha\in E$ is algebraic over $F$, with minimal polynomial $m\in F[x]$, then $[F(\alpha) : F] = \deg m$.
\end{lemma}

\begin{lemma}
\label{lem:adjoin_simple_hom}
Let $E/F$ be a field extension and let $\alpha\in E$ be algebraic over $F$ with minimal polynomial $m\in F[x]$. For any field extension $K/F$, the number of $F$-algebra homomorphisms $F(\alpha)\to K$ is equal to the number of roots of $m$ in $K$.
\end{lemma}

In human mathematics, neither result is hard to prove. One approach is to first construct the isomorphism $F(\alpha) \cong F[x]/(m)$ and then prove the analogous statements for $F[x]/(m)$ in place of $F(\alpha)$. To prove that the degree of $F[x]/(m)$ over $F$ is equal to the degree of $m$, it is enough to show that $\{1,x,\ldots,x^{\deg m - 1}\}$ forms a basis for $F[x]/(m)$ as an $F$-vector space. And one way to show that the number of $F$-algebra homomorphisms $F[x]/(m) \to K$ is equal to the number of roots of $m$ in $K$ is to explicitly construct the bijection between the two sets.

In a couple places in the proofs we have just sketched, it is helpful to invoke the universal property of $F[x]/(m)$. Recall that this universal property is that if $K/F$ is any field extension and $\beta \in  K$ is a root of $m$ then there is a unique $F$-algebra homomorphism $F[x]/(m) \to K$ that sends $x$ to $\beta$.

The formalization of most of these arguments works pretty smoothly in Lean. Working with isomorphisms and universal properties in Lean (and in other proof assistants based on dependent type theory) feels very natural and works well. However, proving that $\{1,x,\ldots,x^{\deg m - 1}\}$ forms a basis for $F[x]/(m)$ turned out to be surprisingly difficult. In our experience, it can be tough to work directly with bases and linear combinations in Lean (partly because it can be tough to work with finite sets and finite lists). 
For this reason, we opted to show that the degree of $F[x]/(m)$ over $F$ is $\deg m$ by constructing an isomorphism between $F[x]/(m)$ and the $F$-vector space of polynomials over $F$ of degree less than $\deg m$.

However, it is sometimes useful to know explicit bases for $F[x]/(m)$ and $F(\alpha)$ (for example, to work with the trace of an element of an algebraic number field). For this reason, Anne Baanen put in the hard work required to actually show that $\{1, x, \ldots, x^{\deg m - 1}\}$ is actually a basis. Along the way, Baanen also defined a new structure called \code{power_basis} to capture situations like $F[x]/(m)$ and $F(\alpha)$ where an $F$-algebra has a basis of the form $\{1, a, a^2, \ldots, a^{n-1}\}$. This has the advantage of allowing general lemmas that cover all of these situations simultaneously. For example, if $A$ and $B$ are both $F$-algebras, $\{1, a, \ldots, a^{n - 1}\}$ is a power basis for $A$, and $b\in B$ is a root of the minimal polynomial of $a$, then there is a unique $F$-algebra homomorphism $A \to B$ mapping $a$ to $b$.
This is another example of the idea that finding the right abstractions is very important in formalizing mathematics, and that sometimes these abstractions don't have any direct analogue in human mathematics.

\section{Galois Theory}

We are now ready to move on to Galois theory.
Recall that we are using separable and normal as our definition of a Galois extension.
The advantage of this is that we can immediately apply the primitive element theorem to Galois extensions.

When proving the Galois correspondence, we won't work directly with the definition of a Galois extension.
Instead, we will rely on the following two lemmas.
The first lemma states that a Galois extension is Galois over any intermediate field.
\begin{lemma}
\label{lem:gal_tower}
Let $E/K/F$ be a tower of field extensions.
If $E/F$ is Galois then so is $E/K$.
\end{lemma}
Recall that \mathlib uses \code{algebra F E} to talk about field extensions, rather than require $F$ to be of type \code{subfield E}.
In the same way, \mathlib uses \code{is_scalar_tower F K E} to talk about towers of field extensions, rather than require $K$ to be of type \code{intermediate_field F E}.
This has the same advantages discussed in section \ref{section:algebra}.

Lemma \ref{lem:gal_tower} is formalized in terms of \code{is_scalar_tower}.
This extra generality is not necessary for the proof of the Galois correspondence.
However, it will be used when proving the equivalent characterizations of Galois extensions.

The second lemma is where we use the primitive element theorem.
\begin{lemma}
\label{lem:gal_aut_card}
Let $E/F$ be a Galois extension of finite degree.
Then $\abs{\Aut(E/F)}=[E:F]$.
\end{lemma}
\begin{proof}
By the primitive element theorem, $E=F(\alpha)$ for some $\alpha\in E$.
Since $E/F$ is of finite degree, $\alpha$ is algebraic over $F$ with minimal polynomial $m\in F[x]$.
By Lemma \ref{lem:adjoin_simple_degree}, $[E:F]=\deg m$.
By Lemma \ref{lem:adjoin_simple_hom}, $\abs{\Aut(E/F)}$ equals the number of roots of $m$ in $E$.
Since $E/F$ is separable and normal, $m$ has $\deg m$ in roots $E$.
Thus, both sides of the desired equality are equal to $\deg m$.
\end{proof}

\subsection{The Galois Correspondence}
\label{section:ftgt}

We will fix a field extension $E/F$.
The Galois correspondence consists of a pair of functions.
For us, this pair of functions will be between the types \code{intermediate_field F E} and \code{subgroup (E ≃ₐ[F] E)}, the type of all subgroups of the group $\Aut(E/F)$ of $F$-algebra automorphisms of $E$.

These maps are defined as usual.
If $H$ is a subgroup of $\Aut(E/F)$, then $E^H$ is the fixed field of $H$.
If $K$ is an intermediate field of $E/F$ then $\Aut(E/K)$ is a subgroup of $\Aut(E/F)$.

There are three numerical facts that we will need.
The first two were already in \mathlib, courtesy of Kenny Lau.
The third is a consequence of Lemmas \ref{lem:gal_tower} and \ref{lem:gal_aut_card} (which uses the primitive element theorem).
\begin{lemma}
\label{lem:inequality}
Let $E/F$ be a field extension of finite degree.
Let $H$ be a subgroup of $\Aut(E/F)$ and let $K$ be an intermediate field of $E/F$.
\begin{enumerate}
    \item $[E:E^H]\leq\abs{H}$ \emph{(}linear independence of characters\emph{)}.
    \item $\abs{\Aut(E/K)}\leq[E:K]$.
    \item If $E/F$ is Galois then $\abs{\Aut(E/K)}=[E:K]$.
\end{enumerate}
\end{lemma}
We can now prove the Galois correspondence.
\begin{theorem}[The Galois Correspondence]
\label{thm:correspondence}
Let $E/F$ be a field extension of finite degree.
Let $H$ be a subgroup of $\Aut(E/F)$ and let $K$ be an intermediate field of $E/F$.
\begin{enumerate}
    \item $\Aut(E/E^H)=H$ and $[E:E^H]=\abs{H}$.
    \item If $E/F$ is Galois then $E^{\Aut(E/K)}=K$.
\end{enumerate}
\end{theorem}
In other words, the functions $H\mapsto E^H$ and $K\mapsto\Aut(E/K)$ define an Galois insertion from subgroups of $\Aut(E/F)$, ordered by reverse inclusion, into intermediate fields of $E/F$.
If $E/F$ is Galois, then this Galois insertion is actually a lattice isomorphism.
\begin{proof}
Combining the first two parts of Lemma \ref{lem:inequality} gives
\[\abs{\Aut(E/E^H)}\leq[E:E^H]\leq\abs{H}.\]
However, $H\leq\Aut(E/E^H)$, so we must have equality.
This proves part 1.
Now assume that $E/F$ is Galois.
Combining the equality $[E:E^H]=\abs{H}$ with the third part of Lemma \ref{lem:inequality} gives
\[[E:E^{\Aut(E/K)}]=\abs{\Aut(E/K)}=[E:K].\]
However, $K\leq E^{\Aut(E/K)}$, so we must have equality.
\end{proof}
Note that the second part of Lemma \ref{lem:inequality} is only needed to avoid assuming that $E/F$ is Galois in the first part of Theorem \ref{thm:correspondence}.
In particular, you can deduce the Galois correspondence from just the first and third parts of Lemma \ref{lem:inequality}.
\subsection{Equivalent Characterizations of Galois Extension}
\label{section:equivalent}
In order for the Galois correspondence to be useful, we will need to prove the equivalent characterizations of Galois extensions.
Let $E/F$ be a finite extension.
We will consider the following equivalent conditions:
\begin{enumerate}
    \item $E/F$ is Galois,
    \item $F$ is the fixed field of $\Aut(E/F)$,
    \item $\abs{\Aut(E/F)}=[E:F]$,
    \item $E$ is the splitting field of a separable polynomial.
\end{enumerate}
When formalizing the equivalence of these statements in Lean, we both formalized specific implications as well as a ``the following are equivalent'' theorem (there is a special syntax for such statements in Lean, which allows us to directly use all $12$ implications without having to prove each one directly).

Before we continue, there is one subtlety with condition (2). Here's how it's stated in Lean:
\begin{lstlisting}
fixed_field (⊤ : subgroup (E ≃ₐ[F] E)) = ⊥
\end{lstlisting}
This looks a bit different than our statement of (2) above. The reason why is that our formalized definition of fixed field was as a function \code{subgroup (E ≃ₐ[F] E) → intermediate_field F E}. So in order to make the types match up in the statement of (2), we must replace $F$ with the smallest intermediate field of $E/F$ and $\Aut(E/F)$ with the largest subgroup of $\Aut(E/F)$. In other words, (2) should really be phrased as
\begin{enumerate}
    \item[$2^\prime$.] The fixed field of the top subgroup of $\Aut(E/F)$ equals the bottom intermediate field of $E/F$.
\end{enumerate}
To a human mathematician, (2) and $(2^\prime)$ look like two different ways of saying the same thing. But to Lean, they are different. 

With this subtlety out of the way, a few directions are easy from what we already know:
\begin{itemize}
    \item $(1)\implies(4)$: Take the minimal polynomial of a primitive element.
    \item $(2^\prime)\iff(3)$: Let $H$ be the top subgroup of $\Aut(E/F)$ in the equality $[E:E^H]=\abs{H}$.
\end{itemize}
It remains to prove $(2^\prime)\implies(1)$ and $(4)\implies(3)$.

For the direction $(2^\prime)\implies(1)$, Kenny Lau had already proved that if $G$ is a finite group of automorphisms of a field $E$ then $E/E^G$ is Galois (in the separable and normal sense).
However, this only tells us that $E$ is Galois over the bottom intermediate field of $E/F$.
Somehow we have to go from knowing that $E/\bot$\footnote{Another technical issue is that $\bot$ needs to be coerced from type \code{intermediate_field F E} to type \code{Type}, but this is handled automatically by Lean.} is Galois to knowing that $E/F$ is Galois.

In human mathematics, this problem seems almost nonsensical: $F$ is just equal to $\bot$.
But in Lean they are not equal because they have different types.
One option would be to define a notion of isomorphic field extensions, prove that $E/F$ is isomorphic to $E/\bot$, and prove that being Galois is preserved along this type of isomorphism.

It turns out that an easier option is to invoke a structure that was useful several times: \code{is_scalar_tower}.
We can apply Lemma \ref{lem:gal_tower} to the tower $E/F/\bot$.
All that's needed is to construct instances \code{algebra ⊥ F} and \code{is_scalar_tower ⊥ F E}.
In other words, we need to construct a ring homomorphism $\bot\to F$ from the bottom intermediate field of $E/F$ to $F$, and show that it is compatible with the existing ring homomorphisms $F\to E$ and $\bot\to E$.

This kind of problem of transporting properties along various types of map (in particular, along isomorphisms) shows up a lot in formalization.
It is possible that transporting properties along isomorphisms could be automated.

Finally, the direction $(4)\implies(3)$ is the trickiest, and requires induction.
In the human mathematics proof, we might consider the roots of $r_1,\ldots,r_n$ of the separable polynomial $p \in F[x]$, and inductively prove that the number of $F$-algebra homomorphisms $F(r_1,\ldots,r_k)\to E$ equals the degree $[F:F(r_1,\ldots,r_k)]$.
The inductive step involves showing that each $F$-algebra homomorphism $F(r_1,\ldots,r_k)\to E$ has exactly $[F(r_1,\ldots,r_{k+1}):F(r_1,\ldots,r_k)]$ extensions to $F(r_1,\ldots,r_{k+1})$.

There are several changes that we make to this proof.
First, rather than enumerating the roots of $p$ and adjoining them one-by-one, we instead apply Lemma \ref{lem:induction1} to the set of roots of $p$ and to the predicate
\[P(K)\iff\text{the number of $F$-algebra homomorphisms $K\to E$ equals $[K:F]$}.\]
Then the inductive step involves counting $F$-algebra homomorphisms $K(\alpha)\to E$, where $K$ is an arbitrary intermediate field of $F/E$ and $\alpha$ is a root of $p$.

In Lean, the easiest way to count something is to establish a bijection.
A bijection in Lean needs to be between two types.
For our bijection, one of the types should be the type of $F$-algebra homomorphisms $K(\alpha)\to E$.
By examining the human mathematics proof, we see that the other type should be the type of pairs $(f,g)$ where $f$ is an $F$-algebra homomorphism $K\to E$ and $g$ is an $F$-algebra homomorphism $K(\alpha)\to E$ extending $f$.
The type of pairs $(f,g)$ is an example of a \textit{dependent type} since the type of $g$ depends on the choice $f$.
Not every type theory theorem prover supports dependent types, so this is one of the advantages of Lean.

We now need to count the number of $g$ extending each $f$.
Ideally, we would reuse Lemma \ref{lem:adjoin_simple_hom}.
This doesn't directly work since $g$ is not a $K$-algebra homomorphism.
However, let $E_f$ be the field $E$ with $K$-algebra structure given by $f$.
Then an $F$-algebra homomorphism $g\colon K(\alpha)\to E$ is the same as a $K$-algebra homomorphism $g\colon K(\alpha)\to E_f$.

Then the inductive step follows from Lemmas \ref{lem:adjoin_simple_degree} and \ref{lem:adjoin_simple_hom}, along with the following proposition.
\begin{proposition}
\label{prop:alg_hom_equiv_sigma}
There is a bijection between $F$-algebra homomorphisms $K(\alpha)\to E$ and pairs $(f,g)$ where $f$ is an $F$-algebra homomorphism $K\to E$ and $g$ is a $K$-algebra homomorphism $K(\alpha)\to E_f$.
Here $E_f$ denotes $E$ with $K$-algebra structure given by $f$.
\end{proposition}
Proposition \ref{prop:alg_hom_equiv_sigma} can be generalized considerably.
In particular, $K(\alpha)/K/F$ can be any \code{is_scalar_tower}.

\section{Future Work: The Abel-Ruffini Theorem}

Though we have formalized a few of the main theorems of Galois theory, there are some important results that we have not formalized yet. We plan to continue developing the field theory library of \mathlib by working on formalizing these remaining results. We, along with Jordan Brown, have recently formalized a proof of the Abel-Ruffini theorem on the insolvability of polynomials of degree five by radicals. This theorem was first formalized a few months earlier in the Coq theorem prover by Sophie Bernard, Cyril Cohen, Assia Mahboubi, and Pierre-Yves Strub \cite{bernard2021unsolvability}.

\section{Acknowledgements}

We would like to thank everyone in the Lean community who helped us to complete this project by answering our questions on the Lean Zulip chat and giving us feedback on our pull requests. Special thanks are due to Anne Baanen and Johan Commelin for numerous suggestions, to Kevin Buzzard for helping solve a few technical problems and to Kenny Lau and Anne Baanen for developing a lot of the field theory in \mathlib on which our project depended. We would also like to thank Johan Commelin, Kevin Buzzard and the anonymous referee for this paper for reading this paper carefully and making many helpful suggestions which greatly improved our exposition. And finally, thanks to all the participants in the Berkeley Lean Seminar during Summer 2020, and especially to Kyle Miller and Jordan Brown.

\bibliographystyle{plain}
\bibliography{bibliography}

\end{document}